\documentclass[11pt]{article}

\usepackage[english]{babel}
\usepackage[utf8x]{inputenc}
\usepackage[T1]{fontenc}
\usepackage[table,xcdraw]{xcolor}
\usepackage{float}
\usepackage{verbatim} 
\usepackage{graphicx}
\usepackage{amsthm}

\begin{document}
\title{A handle is enough for a hard game of Pull}
\author{Oscar Temprano\\
\small{\texttt{oscartemp@hotmail.es}}}
\date{}
\twocolumn[
\begin{@twocolumnfalse}
\maketitle
\begin{abstract}
We are going to show that some variants of a puzzle called pull in which the boxes have handles (i.e. we can only pull the boxes in certain directions) are NP-hard
\end{abstract}
\end{@twocolumnfalse}
]

\newtheorem{teorema}{Theorem}
\newtheorem{lema}{Lemma}

\section{Introduction}

In \cite{Pull}, Marcus ritt proved that two variants of pushpush \cite{Push}, called pull and pullpull, were np-hard. He wondered about the complexity of a variant of pull in which the boxes have either horizontal (left and right) or vertical (up and down) handles.
We will solve that problem in this document, by showing that the problem is NP-hard. We will also show that other variants of pull are NP-hard as well.
In section two we are going to show that pull is NP-hard.
In section three we are going to show that the variant of pull in which the boxes have handles in opposite directions (i.e left and right) is NP-hard.
Finally, in section four, we will show that the variant of pull where the player can only pull boxes in a single direction is NP-hard.
We will show that the previous variants are NP-hard by reducing satisfiability to those variants of pull.
The proofs in this paper are for variants of pull with fixed blocks (Pull-F), but they can be easily adapted so that they work for variants of pull without fixed blocks.
This can be done by filling with blocks all the positions of the board that are inaccesible in the corresponding Pull-F variants, by making the separation between the "gaps" and the boxes inside the crossovers bigger and by making the walls thicker for those gadgets for which it is necessary.
We will represent movable blocks(boxes) by orange/brown squares and fixed blocks by red squares.

\section{Pull is NP-hard}

In this section we are going to prove that Pull is NP-hard, as they did in \cite{Pull}, but we are going to prove the result in a different way. We will adapt the proof  in the next sections to show that other  pull variants are NP-hard

First we are going to show the gadgets that we need to prove that pull is NP-hard.

\subsection{No return gadget}

\begin{figure}[H]
  \centering
    \includegraphics[scale=0.5]{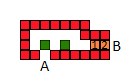}
  \caption{No return gadget (the green squares are places in which the player has to place the boxes in order to exit from the gadget).}
  \label{fig:no_return}
\end{figure}

This is a no return gadget, similar to the ones in \cite{Pull} and \cite{Push}

We can use these gadget for those path variants in which all the boxes can be moved in all directions or those variants in which boxes can only be moved either in horizontal or vertical directions. So we are going to reuse this gadget for the next section.

This gadget allows passage from A to B, forbids passage from B to A, and can only be traversed once

\begin{lema}
\label{one}
The no return gadget can be traversed from A to B , but after it is traversed, te player cannot go back from B to A.
\end{lema}

\begin{proof}
To go to B, the player has to move the boxes 1 and 2 out of the way. In figure 1, the only way to do it is by pulling the boxes to the left (it can be any other direction if the gadget is rotated and/or reflected).
To be able to pull box 2 out of the narrow corridor, the player has to pull  box 1 to the position marked by the leftmost green square, blocking the way to A, the player can then pull the other box to the other green square and exit through B
\end{proof}

\subsection{Unidirectional crossover}

This crossover works like the crossover in \cite{Push} , but it doesn´t matter the order in which the player uses the crossover, he can traverse the crossover horizontally and then vertically or viceversa.
If the player enters the gadget from the entry A1, he can only exit through A2. And if the player enters from B1 he can only exit through B2 (see figure 2).

\begin{figure}[H]
  \centering
    \includegraphics[scale=0.3]{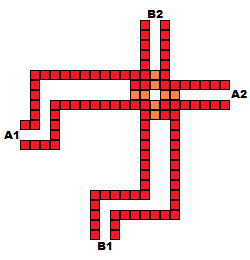}
  \caption{The crossover in its original state}
  \label{crossover_pull_path}
\end{figure}

If the player enters the gadget from A1, he can reach A2 by pulling  to the left all boxes that are blocking the way to A2.

If the player enters the gadget from B1, he can reach B2 by pulling downwards all boxes that are blocking the way to B2.

\begin{figure}[H]
  \centering
    \includegraphics[scale=0.3]{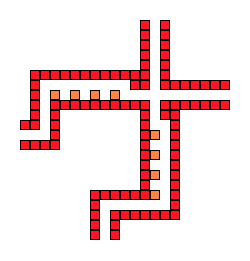}
  \caption{A crossover that has been crossed from both directions}
  \label{crossover_A_2}
\end{figure}

\begin{lema}
\label{two}
In its original state, the crossover gadget can be traversed from A1 to A2 or it can be traversed from B1 to B2. Those are the only possible traversals of the gadget
\end{lema}

\begin{proof}
Let´s suppose that the player enters the gadget from A1, the other direction is simmetrical.
When the player enters, the only thing he can do is pull the two boxes that he finds at the entrance. From there, he cannot go to B2 or B1 because the player has not enough space to pull down the boxes that are blocking the path to B2, and he has not enough space to pull up the boxes that are blocking the way to B1
\end{proof}

Variants of this gadget will form an important component of the different \\crossovers that we will use for the other  variants of pull.

\subsection{Bidirectional crossover}

\begin{figure}[H]
  \centering
    \includegraphics[scale=0.5]{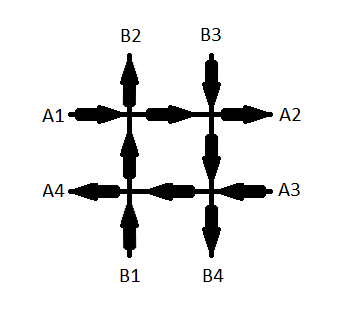}
  \caption{An schematic view of the bidirectional crossover}
  \label{complete_crossover_1}
\end{figure}

This is a crossover that is formed by four limited unidirectional crossovers, rotated and placed as is necessary to make the gadget work. At the entries and exits of the unidirectional crossovers that make the gadget there are no return gadgets.

This gadget is similar to the bidirectional gadget of \cite{Push} and allows the player to traverse it in the horizontal direction and then the vertical direction or viceversa. This gadget only works for the path versions of pull that let the player move the boxes in all directions.

\subsection{Variable gadget}

\begin{figure}[H]
  \centering
    \includegraphics[scale=0.5]{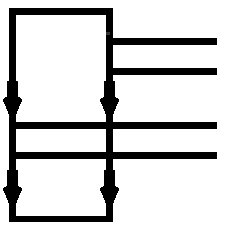}
  \caption{An schematic picture of a variable gadget. The arrows represent no return gadgets inside the gadget. In this example, the variable appears four times in the SAT instance. Two times as a positive literal and other two times as a negative literal. If the player takes the left path, that means that the player has decided to set that variable positive. If the player takes the right path, the player has decided to set the variable negative}
  \label{variable_1}
\end{figure}

For this variant of the game and the following ones, we replace every literal wire inside the variable gadget with two parallel wires. One is used to go to the clause and the other is used to return from the clause. We have to ensure that the player can only use these wires if he comes from the right place.

If the player comes from the path of a positive literal, we don´t have to do any change, since the one way gadget that is just after the literal wires of the negative literal path prevent him from accessing those literal wires.

Now if the player comes from the negative literal path, we have to ensure that the player cannot access any of the literal wires of the positive literal path.

To do that, we simply place an unidirectional crossover for each crossing of the wires inside the variable gadget

\subsection{Clause gadget}

\begin{figure}[H]
  \centering
    \includegraphics[scale=0.5]{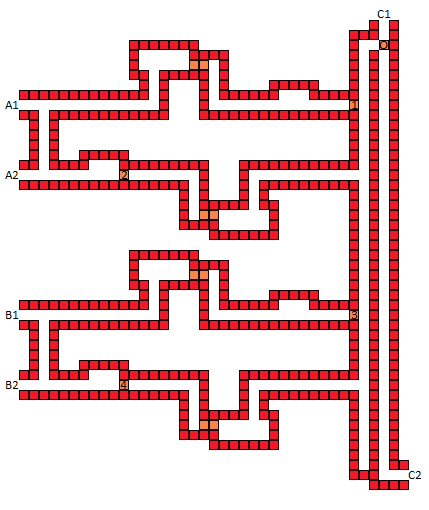}
  \caption{A picture of a clause gadget with two literals, the wires A1 and B1 correspond to entrances to the clause. The wires A2 and B2 are exits of the clause gadget. This gadget works for those variants of pull where the player can pull the boxes in all directions or only in opposite directions (either left and right or up and down)}
  \label{clause_1}
\end{figure}

The gadget has $n + 1$ entrances and $n$ exits where $n$ is the number of literals in the clause.

Every entrance that comes from a literal wire is adjacent to an exit. The entrances and the exits are connected by a corridor 

The gadget has a no return gadget before all entrances and exits, including the C1 entrance and the C2 exit

\begin{lema}
\label{six}
If the player wants to go from C1 to C2 in figure 6, the player had to access the gadget previously from an entrance that belongs to a literal wire
\end{lema}

\begin{proof}
When the player enters the gadget, there is a box (with the label O in figure 6) that it´s blocking the way to B2.
In order to get this box out of the way, the player has to access the gadget from one of the other entrances (A1 or B1) and pull the box one square to the left, so that the box doesn´t block the corridor anymore.
\end{proof}

\begin{lema}
\label{seven}
When the player enters the gadget from an entrance that belongs to a literal wire he has to exit through the adjacent exit
\end{lema}

\begin{proof}
When the player comes from an entrance wire (A1, for instance) he can go to the wire below using the corridor, then he can move the box labelled "1" to the center of the "gap" at his left, because there is a no return gadet at the end of the wire, he has to return to the wire above him and cross a no return gadget in order to get to the clause. Then he has to cross a no return gadget and move the box labelled "3" to the "gap" at the left of the box, then he can enter the gadget and move the box labelled "O" to the left. The no return gadgets that are at the entrance wire (or the boxes at the end of the wires, if the entrances hasn´t been traversed before)  doesn´t allow the player to return from these wires.
If he tries to exit  from another return wire different from the one that he has unlocked (for instance B2) he will get trapped between a no return gadget and a box (with the label "2" in figure 6) in case the wire had not been traversed before, or he won´t be able to go through the wire, because the no return gadget has already been traversed before. He also won´t be able to exit from the wire he entered because a no return gadget is blocking the way.
\end{proof}

\subsection{General view of the \\board}

We have seen the gadgets that we are going to use. Now we are going to shown how to arrange them to create a board.

In the board, the variable gadgets are at the left of the board and the clause gadgets are at the right of the board. Each variable gadget is put below the next one, the same for the clause gadgets. 

\begin{figure}[H]
  \centering
    \includegraphics[scale=0.8]{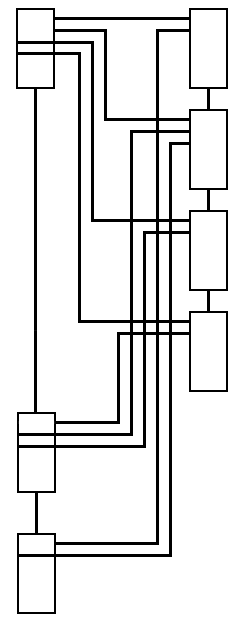}
    \caption{An example board. The variables at the left are $x_1$ $x_2$ and $x_3$. The clauses at the right are $(\neg x_1 \lor \neg x_3) \wedge	(\neg x_1 \lor x_2 \lor x_3) \wedge (x_1 \lor x_2) \wedge (x_1 \lor \neg x_2)$}
  \label{fig:board}
\end{figure}

\begin{teorema}
\label{T2}
Pull is NP-hard
\end{teorema}

\section{Pull where boxes have handles in opposite directions is NP-hard}

We are going now to prove that all variants where boxes have handles in opposite directions (i.e. left and right or up and down) are NP-hard

In this section, and in the following ones, we will only describe the changes that we must do to the gadgets we presented in the previous section so that they work correctly for this variant of pull.

\subsection{Half unidirectional\\ crossover}
This is a crossover that works like the unidirectional crossover of previous section, but it only works in a single direction. That is, it is a crossover that can, for instance, prevent a player from accesing the horizontal direction if the player enters the gadget from the vertical direction. But it can´t prevent the player from traversing it vertically if he has entered the gadget through the horizontal wire 

When the gadget prevents the player from traversing it vertically from inside the horizontal wires, we will say that the gadget is a "vertical crossover"

When the gadget prevents the player from traversing it horizontally from inside the vertical wire, we will say that the gadget is an "horizontal crossover"

\begin{figure}[H]
  \centering
    \includegraphics[scale=0.5]{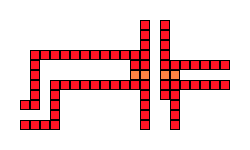}
    \caption{Example picture of an horizontal crossover. If we rotate the gadget 90 degrees, it becomes a vertical crossover}
  \label{fig:horizontal_half_crossover}
\end{figure}

We will use this crossover for all those variants in which we cannot create unidirectional crossovers like the ones from the previous section.This crossover will form a component of bigger crossovers with more functionality

\subsection{Variable gadget}

This gadget is the same variable gagdet the one used in the previous section, but we cannot use the same crossovers we used in the previous section. We have to replace the crossovers of the previous section with half crossovers. See figure 9

\begin{figure}[H]
  \centering
    \includegraphics[scale=0.3]{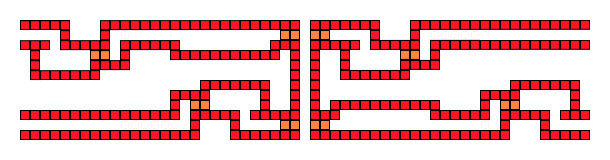}
    \caption{This picture represents the intersection point between an horizontal and a vertical wire inside a variable gadget. The two horizontal wires are a closeup of the two literal wires that the player has to use to go to a clause (up wire) and return from that clause (down wire)}
  \label{fig:var_closeup}
\end{figure}

\subsection{Two way gadget}

This is a gadget that allows the player to go to a position and return back. Once the player returns, the player cant go back to the position he visited previously.

This gadget is formed by two no return gadgets, the arrows in figure 10

\begin{figure}[H]
  \centering
    \includegraphics[scale=0.8]{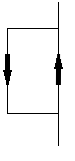}
    \caption{Two way gadget, the direction in which the arrow is pointing represents the direction we must follow to pass through the wire}
  \label{fig:two_way}
\end{figure}

We will only use this gadget inside crossovers

\subsection{Limited Unidirectional Crossover}
This is a unidirectional crossover, similar to the one shown in previous section, but it only works for variants of pull in which the player can only pull boxes in opposite directions. The gadget allows a player to cross from A1 to A2 without allowing the player to exit through B1 or B2 and it allows the player to enter from B1 to B2 without allowing the player to exit through A1 or A2

\begin{figure*}
	\begin{center}
	    \includegraphics[scale=0.8]{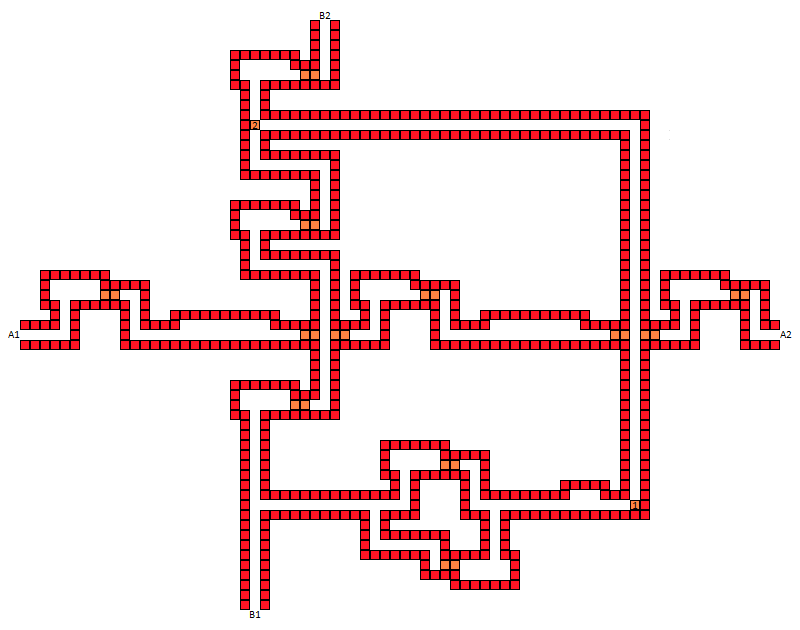}
	    \caption{Complete drawing of the unidirectional crossover. In this example croosover, the player can only pull the boxes left and right.}
	   \label{fig:unidirectional_crossover}
   	\end{center}
\end{figure*}

If the player enters from A1, reaching A2 is straightforward. The player has only to follow the path through A2, going through all the half crossovers and the one way gadgets.

To reach B2 from B1, though, is not so simple. The player can´t simply go upwards until he reaches B2, because he will get trapped between a one way gadget and the block 2. The player must first move through the first corridor to the right of the vertical wire, use it to move the block 2 out of the way upwards and then return to the vertical wire. He can now reach B2 without trouble

\begin{lema}
\label{eleven}
The crossover allows the player to traverse from A1 to A2 without allowing him to reach B1 or B2.
\end{lema}

\begin{proof}
 We observe by looking at the gadget that there are only two intersections in which the horizontal wire crosses the vertical one. We will proof that attempting to enter into the vertical wires by any of the two intersections is impossible or will only get the player trapped inside the gadget.
 
 In the leftmost intersection, we observe that the player can´t enter inside B1 because he is prevented to do so by a one way gadget that cannot be traversed downwards. If the player tries to move upwards, either an already traversed one way gadget or the block 2 will prevent the player from reaching upwards.
 
 In the rightmost intersection, the player is blocked downwards by the block 1 or by two one way gadgets that have already been traversed. This is inevitable, since to pull the block 1 out of the way the player must cross one of the two one way gadgets and then it must  return through the other one, closing it. The player cannot go upwards because there is the block 2 preventing him from reaching B2.
\end{proof}

\begin{lema}
\label{eleven}
The crossover allows the player to traverse from B1 to B2 without leaking to A1 or A2.
\end{lema}

\begin{proof}
 We observe by looking at the gadget that there are only two intersections in which the vertical wires cross the horizontal one. All of them are blocked either by a unidirectional crossover or by a one way gadget that it´s blocked. This is inevitable since we showed in the previous lemma that a player entering through A1 can only exit through A2, going through all the one way gadgets in his way.
\end{proof}

\subsection{Bidirectional crossover}

\begin{figure}[H]
  \centering
    \includegraphics[scale=0.5]{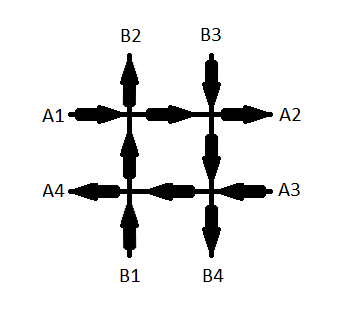}
  \caption{An schematic view of the bidirectional crossover}
  \label{complete_crossover_2}
\end{figure}

This is a crossover that is formed by four limited unidirectional crossovers, rotated and placed as is necessary to make the gadget work. At the entries and exits of the unidirectional crossovers that make the gadget, there are no return gadgets.

This gadget is similar to the bidirectional gadget of \cite{Push} and allows the player to traverse it in the horizontal direction and then the vertical direction or viceversa.

With these components, we can prove the hardness of pull for those versions in which the boxes can only be pulled from opposite directions

\begin{teorema}
\label{thirteen}
The path version of pull in which all boxes can only be pulled from opposite directions is NP-hard
\end{teorema}

\section{Pull with one handle is NP-hard}

In this section, we are going to prove the main claim of our paper. That is, we are going to show that the version of pull where all boxes can only be pulled from one direction (in our examples, the left direction) is NP-hard.

As we did in previous sections, we are only going to describe the \\changes that must be done to the gadgets so that they work for this variant

\subsection{No return gadget}

\begin{figure}[H]
  \centering
    \includegraphics[scale=0.5]{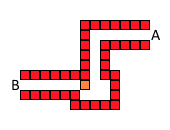}
  \caption{No return gadget for variants in which the player can pull a box in a direction but not the opposite direction.}
  \label{fig:no_return1}
\end{figure}

This gadget is a no return gadget like the one in the previous sections, but adapted for this variant of pull. It allows passage from A to B a single time and once we traverse the gadget we cannot return from B to A.

\subsection{Variable gadget}

We are going to describe the changes we need to make to the variable gadget so that it works for this variant.

If the player comes from the path of a positive literal, we don´t have to do any changes to the gadget, since the one way gadget that is just after the literal wires of the negative literal prevent him from accessing those literal wires.

\begin{figure}[H]
  \centering
    \includegraphics[scale=0.3]{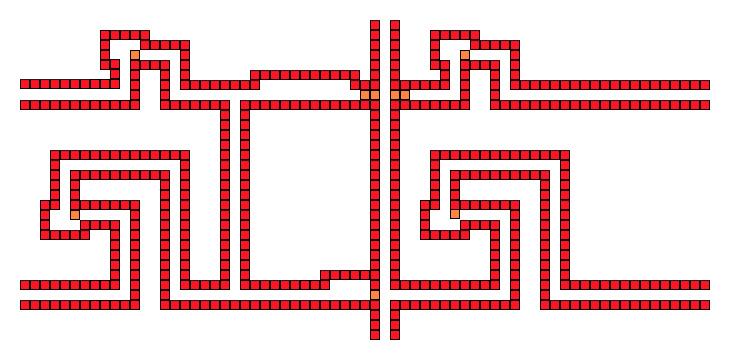}
    \caption{This picture represents the intersection point between an horizontal and a vertical wire inside a variable gadget. The two horizontal wires are a closeup of the two literal wires that the player has to use to go to a clause (up wire) and return from that clause (down wire). In this example figure, the player can only pull the boxes left and down}
  \label{fig:var_closeup_2}
\end{figure}

Now, if the player comes from the negative literal path, we have to ensure that the player cannot access any of the literal wires that belong to the other path.

In the example shown in figure 14, this is done by  placing an horizontal crossover in the wire that goes to the clause. 

For the return wire, a single box prevents the player from accesing the wire. This box can be unlocked from  a corridor that comunicates the two wires. This corridor can only be traversed from the interior of the parallel wire above it.

\subsection{Clause gadget}
This clause gadget is similar the one that appears in the previous, but the no return gadgets have been changed so that they work for this variant of pull.
 
\begin{figure}[H]
  \centering
    \includegraphics[scale=0.5]{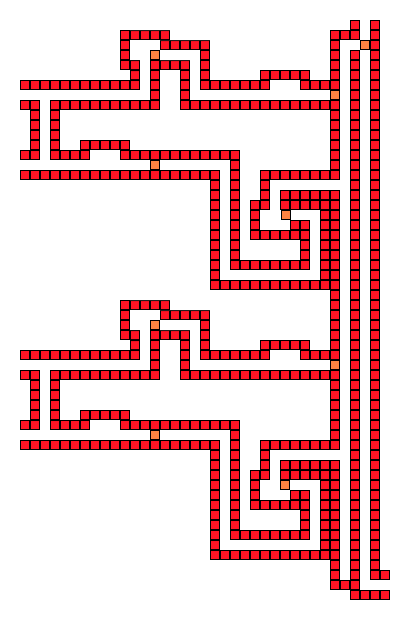}
  \caption{Clause gadget, it is similar to the one shown in section 2 but changing the no return gadgets}
  \label{clause_4}
\end{figure}

\subsection{Crossover gadget}

\begin{figure*}
	\begin{center}
	    \includegraphics[scale=0.8]{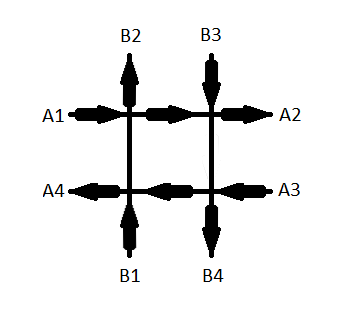}
	    \caption{schematic representation of the crossover.}
	   \label{fig:crossover_scheme_2}
   	\end{center}
\end{figure*}

This is a crossover that is valid for all variants of pull where the player can only pull the boxes from a single direction. In our example, the left direction. This crossover is formed by two "halfs". Each of these halfs contains a vertical wire and the intersection of the two horizontal wires with the vertical wire.

The first half contains a vertical crossover to go towards the clause and the second half contains a vertical wire to return. Of course, the two halfs can be "inverted" (that is, the "left" half can be the "right" half and viceversa)

Now, we will show that the first half of the gadget works as it should by a series of lemmas. The other half works the same, so the proof is symmetrical. The only difference is that the vertical wire is "flipped" vertically in the second half of the gadget.

\begin{lema}
\label{fifteen}
The crossover allows the player to go from A1 to the exit of the first half of the crossover without letting him exit from B1 or B2
\end{lema}

\begin{proof}
After the player crosses the first no return gadget, he has to pull a block in the corridor just below him. This is so that he can return through the wire below him without getting trapped. He can´t use that corridor to reach the other horizontal wire since the block prevents him from reaching it. After crossing another no return gadget and opening the first crossover, he finds the first intersection with the vertical wire. If the player has not traversed the vertical wire before  and he goes upward, after using a no return gadget, he will find himself trapped because there is a block preventing him from going upwards. If, instead, he goes downwards, another no return gadget will block his way. If the player has traversed the vertical wire before, the player will have two no return gadgets, one below him and other above him that will prevent him from going upwards or downwards. 

If, instead, he moves forward, after crossing the second crossover, he will find another intersection with the vertical wire. From here the player can move upwards, but there is an obstacle that blocks his way up. The player can move the obstacle out of the way but he can´t go back to the previous intersection and move upwards from there, since he had to cross a no return gadget to get to the corridor, and that no return gadget prevents him from going to the previous intersection. If the player decides to go down there will be a block preventing him from entering the vertical wire from there. However, if the player traversed the vertical wire before, two no return gadgets will block his way. This is because the player has to use the two way gadget there in order to move the block inside the vertical corridor.

The player is then forced to exit through A2
\end{proof}

\begin{lema}
\label{fifteen}
The crossover allows the player to go from the entrance of the first half of the crossover to A4 without letting the player go to B1 or B2
\end{lema}

\begin{proof}
After crossing the first no return gadget to his left, the player will find himself in the middle of an intersection. If the player has not traversed the vertical wire before, he can go through a no return gadget in his way forwards, then he will find himself trapped between the previous no return gadget and a block above him. If the player has already traversed the vertical wire, a no return gadget will prevent him from going upwards. If the player goes downwards instead, a no return gadget will stand in his way. 

If the player moves left instead, after crossing another no return gadget, he will reach another intersection with the vertical wire. He can´t go into the vertical wire by moving upwards because a block stands in his way. He can´t move this block out of the way and then try to go upwards from the previous intersection because a no return gadget that he has already traversed doesn´t allow him to move right towards the intersection. If the player has crossed the vertical wire before, a used two way gadget will prevent him from reaching the vertical wire. This is because the player has to use that gadget in order to clear the path inside the vertical wire.

The player has to keep moving left and, after traversing another no return gadget, he finds a corridor connecting the two horizontal wires, but because there is a block inside the corridor, and this block can only be pulled to the left, he can´t use the corridor to reach the other horizontal wire. He is then forced to exit through A4.
\end{proof}

\begin{figure*}
    \includegraphics[scale=0.4]{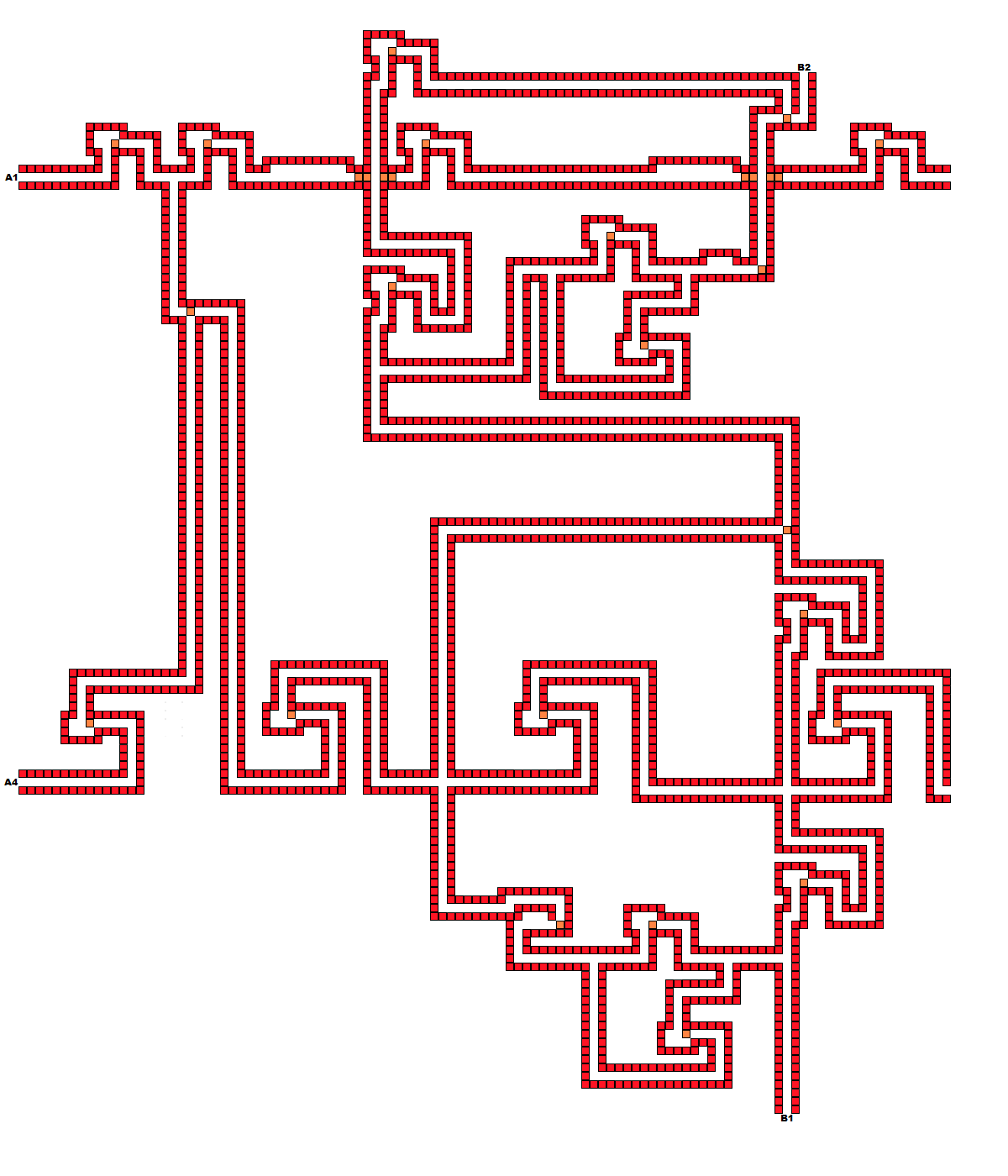}
    \caption{Drawing of the first half of the  crossover. In this example, the player can only pull the boxes to the left }
  \label{fig:horizontal_crossover_one_handle}
\end{figure*}

\begin{figure*}
    \includegraphics[scale=0.4]{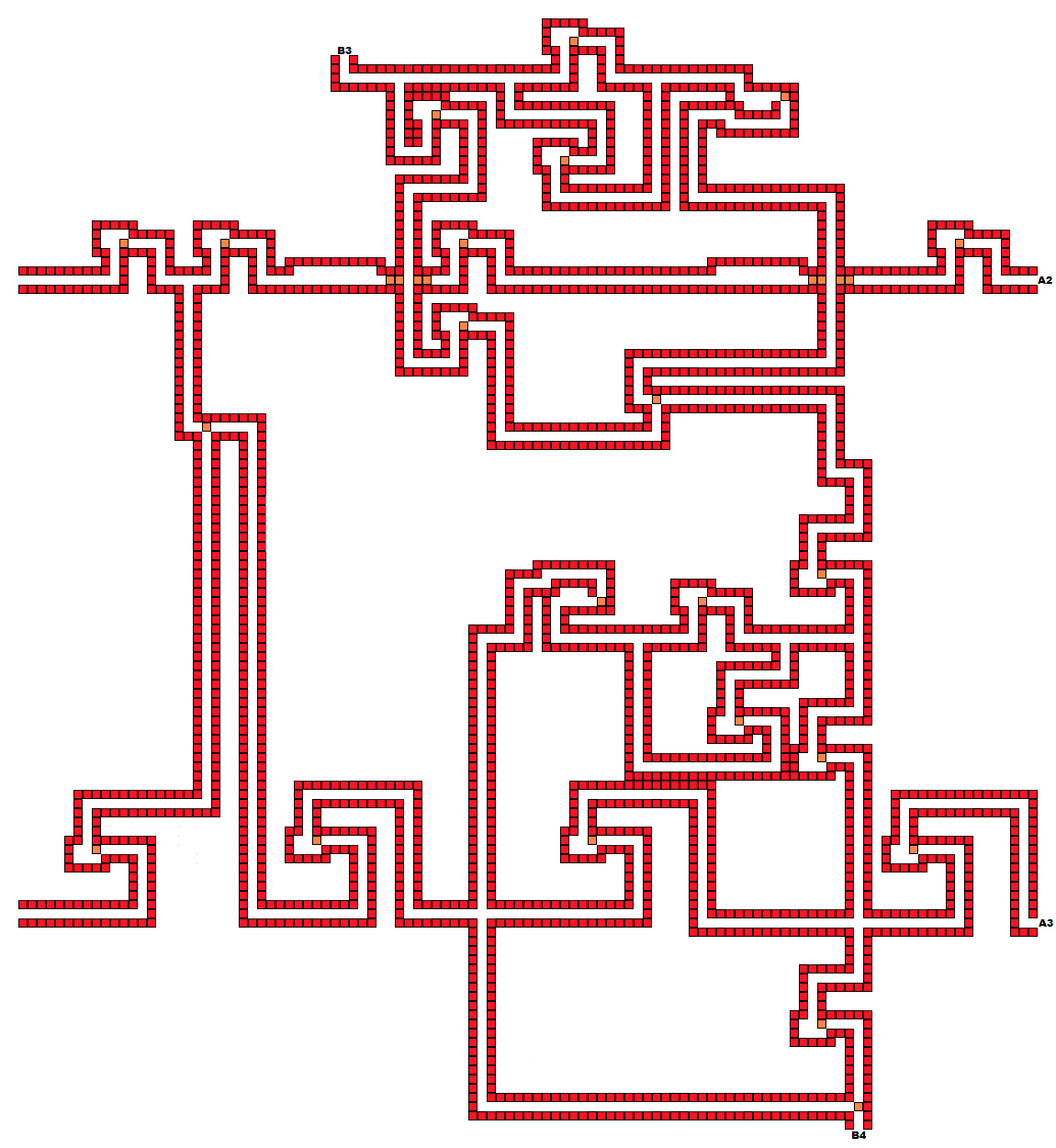}
    \caption{Drawing of the second half of the  crossover.}
  \label{fig:horizontal_crossover_one_handle_2}
\end{figure*}

\newpage

\begin{lema}
\label{fifteen}
The crossover allows the player to go from B1 to B2 without letting the player go to A1 or A4
\end{lema}

\begin{proof}
After crossing the first no return gadget above him, the player will find himself in the middle of an intersection. If the player has not traversed the horizontal wire before, he can go through a no return gadget that is to the left of him, then he can cross another no return gadget to his left, but then, he will find a block that prevents him from moving further left. That block can only be moved out of the way if the player traversed the other horizontal wire before. If the player has already traversed the horizontal wire, a used no return gadget will prevent him from going to the left. If the player goes to the right instead, a no return gadget will stand in his way. 

If the player moves up there are two points where his path crosses with the horizontal wire. But he is prevented from entering the horizontal wire from those intersections, either because there is a crossover that hasn´t been opened or because the two no return gadgets have already been traversed. 

He is then forced to go out through B2.
\end{proof}

\begin{teorema}
\label{sixteen}
The path version of pull in which all boxes can only be pulled from a single direction is NP-hard
\end{teorema}

\section{Open questions and future directions}

Can the results proven here for the path versions of pull be proven for the storage versions of pull?. That is, will it be possible to show that the storage versions of pull where the player can only pull boxes in opposite directions are NP-hard?. Can it be proven that the storage variants of pull where all boxes can only be pulled in only one direction are NP-hard?

We think that the versions where the player can only pull boxes in opposite directions are NP-hard as well. If there exists a variant of planar SAT where the negative clauses are outside the cycle of variables and the positive clauses are inside the cycle of variables, and this variant of SAT is NP-hard, then we can reduce that variant of SAT to the version of storage pull where the player can only pull boxes in opposite directions.

Which variants of Push are NP-hard when the player can only push boxes in certain directions?.

Is the variant of pull where the player cannot revisit a square (called pull-X) still hard if the player can only pull boxes in certain directions?. If the problem is easy this will solve the open question of \cite{Push2} about finding an interesting, but tractable, block-moving puzzle. What about push-X?


\begin{thebibliography}{1}
\bibitem{Pull} Marcus Ritt
{\em “Motion planning with pull moves”}.
http://arxiv.org/abs/1008.2952
\bibitem{Push} Demaine, E. D., Demaine, M. L., and O'Rourke, J.
{\em “PushPush and Push-1 are NP-Hard in 2D”}.
Proc. 12th Canad. Conf. Comput. Geom. (2000), 211-219.
\bibitem{Push2} Erik D. Demaine, Martin L. Demaine, Michael Hoffman and Joseph O'Rourke
{\em “Pushing blocks is hard”}.
Computational Geometry, 26:21-36, 2003
\end{thebibliography}
\end{document}